\documentclass[runningheads, envcountsame, a4paper]{llncs}
\usepackage{makeidx}  
\usepackage{amsmath}
\usepackage{amssymb}
\usepackage{mathtools}
\usepackage{enumitem}
\usepackage{hyperref}
\usepackage{gastex}
\usepackage{pgfplots}
\usepackage{algorithmic}
\usepackage{algorithm}

\usepackage{cite}
\usepackage{graphicx}
\usepackage{array}
\newcolumntype{L}[1]{>{\raggedright\let\newline\\\arraybackslash\hspace{0pt}}m{#1}}
\newcolumntype{C}[1]{>{\centering\let\newline\\\arraybackslash\hspace{0pt}}m{#1}}
\newcolumntype{R}[1]{>{\raggedleft\let\newline\\\arraybackslash\hspace{0pt}}m{#1}}

\usepackage{float}
\usepackage{subfig}

\usepackage{multirow}
\usepackage{rotating}
\usepackage{tikz}
\usetikzlibrary{arrows,automata, positioning}
\definecolor{light-gray}{gray}{0.8}

\providecommand{\fg[1]}{{\color{black}#1}}
\providecommand{\ffg[1]}{{\color{black}#1}}
\providecommand{\jv[1]}{{\color{black}#1}}

\newtheorem{conj}{Conjecture}

\DeclareSymbolFont{rsfscript}{OMS}{rsfs}{m}{n}
\DeclareSymbolFontAlphabet{\mathrsfs}{rsfscript}

\DeclareMathOperator{\rt}{rt}
\DeclareMathOperator{\diam}{diam}

\DeclareMathOperator{\excl}{\mathrm{excl}}
\DeclareMathOperator{\dupl}{\mathrm{dupl}}

\newcommand{\scc}{strongly connected component}
\newcommand{\scn}{strongly connected}

\begin{document}%

\title{On the Interplay\\ Between Babai and \v{C}ern\'y's Conjectures}
\titlerunning{On Babai and \v{C}ern\'y's conjectures}

\authorrunning{F. Gonze et~al.}
\tocauthor{Fran\c cois~Gonze, Vladimir~V.~Gusev, Bal\'azs~Gerencs\'er, Rapha\"el~M.~Jungers, and Mikhail~V.~Volkov}
\toctitle{On the interplay between Babai and \v{C}ern\'y's conjectures}

\author{Fran\c cois Gonze\inst{1} \and Vladimir V. Gusev\inst{1,2} \and Bal\'azs Gerencs\'er\inst{3} \and Rapha\"el M. Jungers\inst{1} \and Mikhail V. Volkov\inst{2}\thanks{Vladimir Gusev and Mikhail V. Volkov were supported by RFBR grant no.\ 16-01-00795, Russian Ministry of Education and Science project no.\ 1.3253.2017, and the Competitiveness Enhancement Program of Ural Federal University. Bal\'azs Gerencs\'er was supported by PD grant no.\ 121107, National Research, Development and Innovation Office of Hungary. This work was supported by the French Community of Belgium and by the IAP network DYSCO. Rapha\"el Jungers is a Fulbright Fellow and a FNRS Research Associate.}}

\institute{ICTEAM Institute, Universit{\'e} Catholique de Louvain, Louvain-la-Neuve, Belgium\\
 \email{$\{$francois.gonze,vladimir.gusev,raphael.jungers$\}$@uclouvain.be}
\and
Ural Federal University, Ekaterinburg, Russia\\
\email{mikhail.volkov@usu.ru}
\and
Alfr\'ed R\'enyi Institute of Mathematics, Budapest, Hungary\\
\email{gerencser.balazs@renyi.mta.hu}
}

\maketitle
\setcounter{footnote}{0}

\begin{abstract}
    Motivated by the Babai conjecture and the \v{C}ern\'{y} conjecture, we study the reset thresholds of automata with the transition monoid equal to the full monoid of transformations of the state set. For automata with $n$ states in this class, we prove that the reset thresholds are upper-bounded by $2n^2-6n+5$ and can attain the value $\tfrac{n(n-1)}{2}$. In addition, we study diameters of the pair digraphs of permutation automata and construct $n$-state permutation automata with diameter $\tfrac{n^2}{4} + o(n^2)$.
\end{abstract}

\section{\fg{Background and Overview}}
\label{sec:intro}

\makeatletter{\renewcommand*{\@makefnmark}{}
\footnotetext{\jv{A short version of this work has been presented at the conference DLT 2017.}}\makeatother}

\emph{Completely reachable automata}, i.e., deterministic finite automata in which every non-empty subset of the state set occurs as the image of the whole state set under the action \ffg{of a suitable input word, appeared in} \fg{the study of descriptional complexity of formal languages~\cite{Maslennikova12} and in relation to the \v{C}ern\'{y} conjecture~\cite{Don16}. In~\cite{BondarVolkov16}} an emphasis has been made on automata in this class with minimal transition monoid size. In the present paper we focus on automata being in a sense the extreme opposites of those studied in~\cite{BondarVolkov16}, namely, on automata of maximal transition monoid size. In other words, we consider automata \emph{with full transition monoid}, i.e., transition monoid equal to the full monoid of transformations of the state set; clearly, automata with this property are completely reachable. There are several reasons justifying special attention to automata with full transition monoid. First, as observed in~\cite{BondarVolkov16}, the membership problem for this class of automata is decidable in polynomial time (of the size of the input automaton) while the complexity of membership in the class of all completely reachable automata still remains unknown. Second, this class contains automata that correspond to Brzozowski's most complex regular languages~\cite{Brzozowski13} and to other regular languages \ffg{that} \fg{play a distinguished role} in descriptive complexity analysis. Finally, and most importantly from our viewpoint, automata with full transition monoid are synchronizing and their synchronization issues constitute a sort of meeting point for two famous open problems---the \emph{Babai conjecture} and the \emph{\v{C}ern\'{y} conjecture}. \ffg{Next, we recall these conjectures and outline the contribution of the present paper in view of these problems.}

\paragraph{\textbf{1.1. The Babai \fg{Conjecture.}}}
Let $A$ be a set of generators of a finite group $G$. The \emph{Cayley graph} $\Gamma(G,A)$ consists of $G$ as the set of vertices and the edges $\{g, ga\}$ for all $g \in G$, $a \in A$. The \emph{diameter} of $\Gamma(G,A)$ is the maximum among the lengths of shortest paths between any two vertices. In group theory terms, the diameter of $\Gamma(G,A)$ is the smallest $\ell$ such that every $g \in G$ can be represented as $g=a_1^{\varepsilon_1} a_2^{\varepsilon_2} \cdots a_{\ell}^{\varepsilon_\ell}$, where $\varepsilon_i \in \{1,-1\}$ and $a_i \in A$ for all $i =1,\dots, \ell$. The \emph{diameter} $\diam(G)$ of $G$ is the maximal diameter of $\Gamma(G,A)$ among all generating sets $A$ of $G$. The notion of group diameter is related to the growth rate in groups, expander graphs, random walks on groups and their mixing times, see, e.g.,~\cite{Saloff-Coste2004, MR3348442}. Recently, the following conjecture received significant attention:

\begin{conj}[\mdseries Babai~\cite{BaSe92}]
The diameter of each non-abelian finite simple group $G$ does not exceed $(\log |G|)^{O(1)}$, where the implied constant is absolute.
\end{conj}
Note that for the case of the symmetric group $S_n$, this conjecture readily implies $\diam (S_n)\le n^{O(1)}$. (The group $S_n$ is not simple but for $n\ge5$ it contains a non-abelian simple subgroup of index 2.)

The Babai conjecture was proved for various classes of groups, but despite intensive research effort it remains open, see~\cite{HeSe14} for an overview. In the case of $S_n$, a recent breakthrough gives only a quasipolynomial upper bound, namely, $\exp(O((\log n)^4 \log \log n)$, and it relies on the Classification of Finite Simple Groups~\cite{HeSe14}. It is even more astonishing if we compare it to the best known lower bound in this case: for the classical set of generators consisting of the transposition $(1,2)$ and the full cycle $(1,2,\dots, n)$, every permutation in $S_n$ can be expressed as a product of at most $\sim \tfrac{3 n^2}{4}$ (asymptotically) generators~\cite{Zubov}.

\paragraph{\textbf{1.2. The \v{C}ern\'{y} \fg{Conjecture.}}}
\setcounter{footnote}{0}
Recall that a deterministic finite state automaton (DFA) is a triple\footnote{As initial and final states play no role in our considerations, we omit them.} $\langle Q, \Sigma, \delta \rangle$, where $Q$ is a finite set of states, $\Sigma$ is a finite set of input symbols called the \emph{alphabet}, and $\delta$ is a function $\delta\colon Q \times \Sigma \rightarrow Q$ called the \emph{transition function}. A \emph{word} is a sequence of letters from the alphabet. The \emph{length} of a word is the number of its letters. We can look at $\delta(q,a)$ as the result of the \emph{action} of the letter $a\in\Sigma$ at the state $q \in Q$. We extend this action to the action of words over $\Sigma$ on $Q$ denoting, for any word $w$ and any state $q\in Q$, the state resulting in successive applications of the letters of $w$ from left to right by $q{\cdot}w$. For a subset $P\subseteq Q$, we write $P{\cdot}w$ for the set $\{p{\cdot}w\mid p\in P\}$.

A DFA $\mathrsfs{A}=\langle Q, \Sigma, \delta \rangle$ is called \emph{synchronizing} if there exist a word $w$ and a state $f$ such that $Q{\cdot}w=\{f\}$. Any such word is called a \emph{synchronizing} or \emph{reset} word. The minimum length of reset words for $\mathrsfs{A}$ is called the \emph{reset threshold} of $\mathrsfs{A}$ and is denoted by $\rt(\mathrsfs{A})$. Synchronizing automata appear in various branches of mathematics and are related to synchronizing codes~\cite{berstel2009codes}, part orienting problems~\cite{Natarajan:1986,Natarajan:1989}, substitution systems~\cite{Frettloh&Sing:2007}, primitive sets of matrices~\cite{DBLP:journals/corr/GerencserGJ16}, synchronizing groups~\cite{ArCaSt15}, convex optimization~\cite{gonze2016synchronizing}, and consensus theory~\cite{PYChev}.

\begin{conj}[\mdseries \v{C}ern\'{y}~\cite{Cerny1964,CernyPirickaRosenauerova1971}]
The reset threshold of an $n$-state synchronizing automaton is at most $(n-1)^2$.
\end{conj}
If the conjecture holds true, then the value $(n-1)^2$ is optimal, since for every $n$ there exists an $n$-state automaton $\mathrsfs{C}_n$ with the reset threshold equal to $(n-1)^2$~\cite{Cerny1964}.

The \v{C}ern\'{y} conjecture has gained a lot of attention in automata theory. It has been shown to hold true in various special classes\cite{Dubuc98,Kari03,Ry97,STEINBERG20115487,GrechK13a,Steinberg:2011}, but in the general case, it remains open for already half a century. For more than 30 years, the best upper bound was $\tfrac{n^3 - n}{6}$, obtained in~\cite{Pin83a, Frankl82} and independently in~\cite{KlyachkoRystsovSpivak87}. Recently, a small improvement on this bound has been reported in~\cite{Szykula2017}: the new bound is still cubic in $n$ but improves the coefficient $\frac16$ at $n^3$ by $\frac4{46875}$. A survey on synchronizing automata and the \v{C}ern\'{y} conjecture can be found in~\cite{volkov_survey}.

In order to make the relationship between the \v{C}ern\'{y} and the Babai conjectures more visible, we borrow from~\cite{AV2004} the idea of restating the former in terms similar to those used in the formulation of the latter. Let $T_n$ be the full transformation monoid of an $n$-element set $Q$. A transformation $t \in T_n$ is a \emph{constant} if there exists $f \in Q$ such that for all $q \in Q$ we have $t(q)=f$. We can state the \v{C}ern\'{y} conjecture as follows: for every set of transformations $A \subseteq T_n$, if the submonoid generated by $A$ contains a constant, then there exists a constant $g$ such that $g=a_1 a_2\cdots a_\ell$, where $\ell \leq (n-1)^2$ and $a_i \in A$ for all $i=1,\ldots, \ell$. It is easy to see that this formulation is equivalent to the original one by treating the letters of an automaton as the transformations of its state set since reset words precisely correspond to constant transformations.

\paragraph{\textbf{1.3. Our \fg{Contributions.}}}
The first part of our paper is devoted to the following \emph{hybrid Babai--\v{C}ern\'{y} problem}\footnote{\ffg{During the preparation of this paper we discovered
that the same question was also posed in~\cite[Conjecture 3]{SALOMAA2003},
though its connection with Babai's problem was not
registered there.}}: given a set of generators $A$ of the full transformation monoid $T_n$, what is the length $\ell(A)$ of the shortest product $a_1a_2 \cdots a_\ell$ with $a_i\in A$ which is equal to a constant? Namely, we are interested in the bounds on $\ell(A)$ that depend only on $n$. The hybrid Babai--\v{C}ern\'{y} problem is a special case of the \v{C}ern\'{y} problem. Indeed, it is a restriction to the class of DFAs with the \emph{transition monoid}, i.e., the transformation monoid generated by the actions of letters, equal to $T_n$. Of course, the general cubic upper bound is valid, but not the lower bound, since the \v{C}ern\'{y} automata $\mathrsfs{C}_n$  do not belong to this class (even though they are completely reachable, see~\cite{BondarVolkov16}). In Section~\ref{sec:automata_transitionTn} we establish that the growth rate of $\ell(n)$ is $\Theta(n^2)$, more precisely, we show that $\frac{n(n-1)}{2}\le\ell(n)\le 2n^2 -6n +5$. We also present the exact values of $\ell(n)$ for small values of $n$ resulting from our computational experiments. Our contribution can be also seen as a progress towards resolution of Conjecture~3 from~\cite{SALOMAA2003}.

The second part of our paper is devoted to a ``local'' version of the Babai problem where we restrict our attention to the action on the set of (unordered) pairs. Let $A$ be a set of permutations from $S_n$. The \emph{pair digraph} $P(A)$ consists of pairs $\{i,j\}$ as the set of vertices and the edges $(\{i, j\}, \{ia, ja\})$ for all $i,j$ and $a \in A$. The \emph{diameter} of $P(A)$, denoted $\diam P(A)$, is the maximum among the lengths of shortest (directed) paths between any two vertices. We study the behavior of $\diam P(A)$ in terms of $n$. The problem comes from analysis of certain aspects of Markov chains and group theory~\cite{Frid98}, but our interest in it is mainly motivated by its importance for the theory of synchronizing automata. Indeed, every synchronizing automaton $\mathrsfs{A}$ must have a letter $a$, say, whose action merges a pair of states. Thus, one can construct a reset word for $\mathrsfs{A}$ by successively moving pairs of states to a pair merged by $a$. If $\mathrsfs{A}$ possesses sufficiently many letters acting as permutations (as automata with the full transition monoid do), one can move pairs by these permutations, and hence, upper bounds on the diameter of the corresponding pair digraph induce upper bounds on $\rt(\mathrsfs{A})$.

Clearly, $\diam P(A) \leq \tfrac{n(n-1)}2$ for all $A \subseteq S_n$. In Section~\ref{sec:2trans} we establish the lower bound $\tfrac{n^2}{4} + o(n^2)$ on $\diam P(A)$ by presenting a series of examples with only two generators for every odd $n$.


\paragraph{\textbf{1.4. Related \fg{Work.}}}
The diameters of groups and semigroups constitute a relatively well studied topic. A general discussion on diameters and growth rates of groups can be found in~\cite{MR3348442}. Various results about the diameter of $T_n$ and its submonoids are described in~\cite{SALOMAA2003, Panteleev2015}. The length of the shortest representation of a constant (including the case of partially defined transformations) is typically studied in the framework of synchronizing automata, see~\cite{volkov_survey,DBLP:journals/corr/Vorel14,ananichev2013slowly}.

\section{Automata \fg{with Full Transition Monoid}}
\label{sec:automata_transitionTn}

\paragraph{\textbf{2.1. Na{\"i}ve \fg{Construction.}}}
Recall that, on the one hand, the \v{C}ern\'y automata $\mathrsfs{C}_n$ from~\cite{Cerny1964} have two letters of which one acts as a cyclic permutation and the other fixes all states, except one, which is mapped to the next element in the cyclic order defined by the cyclic permutation. On the other hand, the extremal case of the Babai conjecture for $S_n$ is composed of a cyclic permutation and the transformation which fixes all elements except two, which are neighbors in the cyclic order defined by the cyclic permutation. Therefore, one could wonder if a combination of these transformations could result in a DFA with both large reset threshold and full transition monoid.

\fg{
\begin{figure}[ht]
\begin{center}
\scalebox{0.8}{
\begin{tikzpicture}[->,>=stealth',shorten >=1pt,auto,node distance=2cm,semithick]
  \tikzstyle{every state}=[fill=light-gray,text=black, scale=1]

  \node[state] (A)                    {$q_n$};
  \node[state]         (B) [right of=A] {$q_1$};
  \node[state]         (C) [right of=B] {$q_2$};
  \node[state]         (D) [right of=C] {$q_3$};

  \node         (E) [below right  of=D] {$\cdots$};

  \node[state]         (F) [below left of=E] {$q_{k-1}$};
  \node[state]         (G) [left of=F] {$q_{k}$};
  \node[state]         (H) [left of=G] {$q_{k+1}$};
  \node[state]         (I) [left of=H] {$q_{k+2}$};

  \node         (J) [above left  of=I] {$\cdots$};

  \path (A) edge [loop below]  	node {$b,c$} (A)
  			edge             	node {$a$} (B)
        (B) edge              	node {$a,b$} (C)
            edge [loop below]  	node {$c$} (B)
        (C) edge              	node {$a$} (D)
        	edge [loop below] 	node {$b,c$} (D)
        (D) edge  			  	node {$a$} (E)
        	edge [loop below]  	node {$b,c$} (A)
		(E) edge              	node {$a$} (F)
        (F) edge              	node {$a$} (G)
        	edge [loop above] 	node {$b,c$} (D)
        (G) edge [bend left=10] node {$a,c$} (H)
        	edge [loop above] 	node {$b$} (D)
        (H) edge              	node {$a$} (I)
        	edge [bend left=10]	node {$c$} (G)
        	edge [loop above] 	node {$b$} (D)
        (I) edge              	node {$a$} (J)
        	edge [loop above] 	node {$b,c$} (D)
        (J) edge              	node {$a$} (A);
\end{tikzpicture}}
\end{center}
\caption{The automaton $\mathrsfs{C\!B}_{n,k}$}\label{MixedCernyBabai}
\end{figure}
}

The construction is defined as follows. There are $n$ states $q_1,\dots,q_n$ and three letters $a$, $b$, and $c$. The letter $a$ acts as a cyclic permutation on the states, following their indices. The letter $b$ fixes all states, except $q_1$, which is mapped to $q_2$ by $b$. The letter $c$ fixes all states, except $q_k$ and $q_{k+1}$, for some $k$, which are swapped by $c$. The resulting automaton $\mathrsfs{C\!B}_{n,k}$ is shown in Fig.~\ref{MixedCernyBabai}. We notice that if we remove the letter $c$, we obtain the automaton $\mathrsfs{C}_n$ from the \v{C}ern\'y family providing the largest currently existing lower bound in the \v{C}ern\'y problem, and if we remove the letter $b$, we obtain a generating set of the group $S_n$ providing the largest currently existing lower bound in the Babai problem for $S_n$. Also observe that in the case where $k=2$, our automaton is nothing but Brzozowski's ``Universal Witness''~\cite{Brzozowski13} recognizing the most complex regular language, i.e., the language witnessing at once practically all tight lower bounds found for the state complexity of various operations with regular languages, see~\cite[Theorem~6]{Brzozowski13}.
The next result shows that, however, the reset threshold of the automaton $\mathrsfs{C\!B}_{n,k}$ is upper-bounded by $O(n\log n)$, while, as we show later, among automata with full transition monoid there exist ones whose reset threshold is a quadratic function of their state number.

\begin{theorem}
\label{thm:naive}
The automaton $\mathrsfs{C\!B}_{n,k}$ has a reset word of length at most $4n\lceil\log_2n\rceil$.
\end{theorem}
\jv{
\begin{proof}
Recall that we aim to show that, for each $k$, the automaton $\mathrsfs{C\!B}_{n,k}$ has a reset word of length at most $4n\lceil\log_2n\rceil$. It is easy to see that the word $b(cab)^{n-2}$ of length $3n-5<4n\lceil\log_2n\rceil$ resets the automaton $\mathrsfs{C\!B}_{n,1}$, so that we assume that $k>1$ in the rest of the proof.

We construct a word $w$ letter-by-letter in several rounds, starting with the empty word. The main parameter in our construction is the current image of the state set of $\mathrsfs{C\!B}_{n,k}$ under the action of the word constructed so far; let $S$ stand for this image. (Thus, we have $S=\{q_1,\dots,q_n\}$ at the beginning, and $S$ becomes a singleton at the end of the process.) It is quite helpful to visualize $S$ as the set whose states bear certain tokens. If one colors states covered by tokens light-gray, then Fig.~\ref{MixedCernyBabai} represents the initial position while Fig.~\ref{fig:claim2} shows some intermediate situation.
\begin{figure}[h]
\begin{center}
\scalebox{0.9}{
\begin{tikzpicture}[->,>=stealth',shorten >=1pt,auto,node distance=2cm,semithick]
  \tikzstyle{every state}=[text=black, scale=1]

  \node[state,fill=light-gray] (A)                    {$q_n$};
  \node[state]         (B) [right of=A] {$q_1$};
  \node[state,fill=light-gray]         (C) [right of=B] {$q_2$};
  \node[state,fill=light-gray]         (D) [right of=C] {$q_3$};

  \node         (E) [below right  of=D] {$\cdots$};

  \node[state,fill=light-gray]         (F) [below left of=E] {$q_{k-1}$};
  \node[state]         (G) [left of=F] {$q_{k}$};
  \node[state,fill=light-gray]         (H) [left of=G] {$q_{k+1}$};
  \node[state,fill=light-gray]         (I) [left of=H] {$q_{k+2}$};

  \node         (J) [above left  of=I] {$\cdots$};

  \path (A) edge [loop below]  	node {$b,c$} (A)
  			edge             	node {$a$} (B)
        (B) edge              	node {$a,b$} (C)
            edge [loop below]  	node {$c$} (B)
        (C) edge              	node {$a$} (D)
        	edge [loop below] 	node {$b,c$} (D)
        (D) edge  			  	node {$a$} (E)
        	edge [loop below]  	node {$b,c$} (A)
		(E) edge              	node {$a$} (F)
        (F) edge              	node {$a$} (G)
        	edge [loop above] 	node {$b,c$} (D)
        (G) edge [bend left=10] node {$a,c$} (H)
        	edge [loop above] 	node {$b$} (D)
        (H) edge              	node {$a$} (I)
        	edge [bend left=10]	node {$c$} (G)
        	edge [loop above] 	node {$b$} (D)
        (I) edge              	node {$a$} (J)
        	edge [loop above] 	node {$b,c$} (D)
        (J) edge              	node {$a$} (A);
\end{tikzpicture}}
\end{center}
\caption{Tokens mark a subset in $\mathrsfs{C\!B}_{n,k}$}\label{fig:claim2}
\end{figure}
When a letter $x\in\{a,b,c\}$ is applied to $S$, the token that covers $q_i$, say, moves to the state $q_i{\cdot}x$; in more visual terms, the token ``slides'' along the arrow representing the transition $q_i\to q_i{\cdot}x$. If two tokens arrive at the same state, which happens whenever both $q_1$ and $q_2$ bear tokens and the letter $b$ is applied, we remove one of the tokens.

In the course of our construction, rounds of two sorts alternate: \emph{merging}, in which only $a$'s and $b$'s are applied to $S$, and \emph{pairing}, in which only $a$'s and $c$'s are applied. We call a state from $S$ \emph{isolated} if both its neighbor states (with respect to the cyclic order defined by $a$) are not in $S$. A merging round starts whenever $|S|>1$ and $S$ has at most one isolated state, and it lasts while $S$ contains non-isolated states; a pairing round starts whenever $|S|>1$ and all states in $S$ are isolated, and it lasts while $S$ contains more than one isolated state. Every round consists of a number of steps, in each of which we choose a letter, append the chosen letter to the word $w$ and update the set $S$ by applying the letter to it. The choice is done according to one of the two following rules (M) and (P) used during merging and pairing rounds, respectively:
\begin{itemize}
\item[(M)] $b$ is chosen whenever $q_1,q_2\in S$; otherwise $a$ is chosen;
\item[(P)] $c$ is chosen whenever $q_{k+1}\in S$, but $q_k,q_{k+2}\notin S$ (so that $q_{k+1}$ is isolated); otherwise $a$ is chosen.
\end{itemize}

Clearly, at the beginning no state is isolated, and hence, the first round of our construction must be merging. It amounts to an immediate calculation to see that by the end of the first round, we have $w=b(a^2b)^{\lfloor\frac{n-1}2\rfloor}$ and $S=\{q_2,q_4,\dots q_{2\lfloor\frac{n}2\rfloor}\}$.

Now we are going to verify two claims.

\emph{Claim 1.} If $|S|=m$ before any of the next merging rounds, then $|S|=\lceil\frac{m}2\rceil$ at the end of the round.

We say that two neighbor states $q_\ell,q_\ell{\cdot}a\in S$ form an \emph{isolated couple} if each of these states has exactly one neighbor in $S$.

\emph{Claim 2.} If $|S|=m$ before a pairing round, then at the end of the round $S$ is partitioned in either $\frac{m}2$ isolated couples (if $m$ is even) or $\frac{m-1}2$ isolated couples and one isolated state (if $m$ is odd).

First we prove Claim 2. The \emph{distance from $q_i$ to $q_j$} is $\min\{d\in\mathbb{N}\mid q_i{\cdot}a^d=q_j\}$. We order tokens that cover the states in $S$ according to the distance from their states to the state $q_{k+1}$: the $i$-th token is the one that covers the state with the distance $d_i$ to $q_{k+1}$, where $0\le d_1<d_2<\dots<d_m<n$. Now consider the evolution of the set $S$ under the choice of letters according to the rule (P). Clearly, the first $d_1$ choices are all $a$'s. After that the first token reaches $q_{k+1}$. Since the action of $a$ translates the set $S$, without affecting distances between its states, all states in $S$ remain isolated at this point. In particular, $q_{k+1}$ is isolated, and hence, (P) forces the letter $c$ to be applied. This moves the first token ``backwards'' to the state $q_k$ while all other tokens keep their positions. The next letter to be applied is $a$, and its application moves all tokens one step ``forwards'' so that the token from $q_k$ returns to $q_{k+1}$. Clearly, the distance from the state that holds the second token to $q_{k+1}$ becomes $d_2-d_1-1$ after these two moves. If the state $q_{k+1}$ remains isolated, another application of $c$ is invoked, followed by another application of $a$, and this results in a further decrement of the distance from the state that holds the second token to $q_{k+1}$. Eventually, after the suffix $a^{d_1}(ca)^{d_2-d_1}$ is appended to $w$, the second token reaches the state $q_k$. At this moment, the third token (if it exists) covers a state with distance $d_3-d_2>1$ to $q_k$ whence $q_{k+1},q_k$ form an isolated couple in $S$. The two tokens covering these states will then remain adjacent till the end of the round.

If $m=2$ or $m=3$, we are done. If $m>3$, we proceed in the same way. Namely, the next $d_3-d_2$ choices are all $a$'s. After that the third token reaches $q_{k+1}$. Except the first two, all other tokens remain isolated. Now (P) forces $c$ and $a$ to be alternatively chosen $d_4-d_3$ times each. This makes the third token shuffle between $q_{k+1}$ and $q_k$, while the fourth and the next tokens move $d_4-d_3$ steps ``forwards''. After that $q_{k+1},q_k$ form yet another isolated couple in $S$, etc.

We have shown that at the end of the round, the set $S$ indeed consists of either $\frac{m}2$ isolated couples (if $m$ is even) or $\frac{m-1}2$ isolated couples and one isolated state (if $m$ is odd). Moreover, the suffix appended to $w$ during the round is of the form
\begin{equation}
\label{eq:suffix2}
a^{d_1}(ca)^{d_2-d_1}a^{d_3-d_2}(ca)^{d_4-d_3}a^{d_5-d_4}(ca)^{d_6-d_5}\cdots.
\end{equation}
The letter $a$ occurs in this suffix $d_m$ times if $m$ is even and $d_{m-1}$ times if $m$ is odd, and the number of occurrences of $c$ is less than that of $a$. Since $d_{m-1}<d_m<n$, we conclude that the length of the suffix \eqref{eq:suffix2} is less than $2n$.

\smallskip

Now it is easy to prove Claim~1. In view of Claim~2, at the beginning of the round, the set $S$ consists of either $\frac{m}2$ isolated couples (if $m$ is even) or $\frac{m-1}2$ isolated couples and one isolated state (if $m$ is odd). If $\{q_\ell,q_\ell{\cdot}a\}$ is an isolated couple, we say that $q_\ell$ is its \emph{left state}. Now we order isolated couples in $S$ according to the distance from their left states to the state $q_1$: the $i$-th couple is the one with the distance $d_i$ from its left state to $q_1$, where $0\le d_1<d_2<\dots<d_{\lceil\frac{m}2\rceil}<n$. Consider the evolution of the set $S$ under the choice of letters according to the rule (M). The first $d_1$ choices are all $a$'s. After that the tokens that initially covered the states of the first isolated couple arrive at the states $q_1$ and $q_2$, and hence, (M) forces the letter $b$ to be applied. This application removes the token from $q_1$ and does not change anything else. The state $q_2$ then becomes isolated. The next $d_2-d_1$ choices are again all $a$'s, and the successive applications of these $a$'s bring tokens that initially covered the states of the second isolated couple to the states $q_1$ and $q_2$. Then, again, $b$ is chosen, removing the token from $q_1$ and creating yet another isolated state in $S$, etc. At the end of the round, exactly one token from each isolated couple is removed and all remaining states are isolated. The number of these states is $\frac{m}2$ if $m$ is even or $\frac{m+1}2$ if $m$ is odd; in short, $\lceil\frac{m}2\rceil$, as claimed.

Moreover, the suffix appended to $w$ during the round is of the form
\begin{equation}
\label{eq:suffix1}
a^{d_1}ba^{d_2-d_1}b\cdots a^{d_{\lceil\frac{m}2\rceil}-d_{\lceil\frac{m}2\rceil-1}}b.
\end{equation}
The letter $a$ occurs in this suffix $d_{\lceil\frac{m}2\rceil}<n$ times and the letter $b$ occurs $\lceil\frac{m}2\rceil<n$ times, whence the length of the suffix \eqref{eq:suffix1} is less than $2n$.

\smallskip

Claim~1, together with the observation we made about the first merging round, readily implies that the number of merging rounds is at most $\lceil\log_2n\rceil$. Since merging and pairing rounds alternate, the total number of rounds is upper-bounded by $2\lceil\log_2n\rceil$. As observed after the proofs of Claims~1 and~2, a suffix of length less than $2n$ is appended to the current word $w$ during each round. Clearly, at the end of the process, $w$ becomes a reset word for $\mathrsfs{C\!B}_{n,k}$, and by the construction the length of $w$ is less than $2n\cdot 2\lceil\log_2n\rceil=4n\lceil\log_2n\rceil$. \qed

\end{proof}
}

\paragraph{\textbf{2.2. Random \fg{Sampling and Exhaustive Search.}}}
Every DFA with the transition monoid $T_n$ necessarily has permutation letters that generate the whole symmetric group $S_n$ and a letter of rank $n-1$ (i.e., a letter whose image has $n-1$ elements). It is a well known fact that the converse is also true, i.e., the transition monoid of any automaton with permutation letters generating $S_n$ and a letter of rank $n-1$ is equal to $T_n$, see, e.g., \cite[Theorem~3.1.3]{Ganyushkin&Mazorchuk:2009}.

Relying on a group-theoretic result by Dixon~\cite{Dixon69}, Cameron ~\cite{Cam13} observed that an automaton formed by two permutation letters taken uniformly at random and an arbitrary non-permutation letter is synchronizing with high probability.  We give an extension by using another non-trivial group-theoretical result, namely, the following theorem by Friedman \emph{et al.}~\cite{Frid98}:
\begin{theorem}
\label{thm:friedman}

For every $r$ and $d\ge2$ there is a constant $C$ such that \fg{ for $d$ permutations $\pi_1,\pi_2,\dots,\pi_d$ of $S_n$ taken uniformly at random, the following property $\mathrm{F}_r$ holds with probability tending to 1 as $n\to\infty$}: for any two $r$-tuples of distinct elements in $\{1,2,\dots,n\}$, there is a product of less than $C\log n$ of the $\pi_i$'s \fg{which maps the first} $r$-tuple to the second.

\end{theorem}

\begin{corollary}
\label{cor:random}
There is a constant $C$ such that the reset threshold of an $n$-state automaton with two random permutation letters and an arbitrary non-per\-mu\-ta\-tion letter does not exceed $Cn\log n$ with probability that tends to 1 as $n\to\infty$.
\end{corollary}

\begin{proof}
Let $\mathrsfs{A}=\langle Q, \Sigma, \delta \rangle$ stand for the automaton in the formulation of the corollary. We let $a\in\Sigma$ be the non-permutation letter and assume that the two permutation letters in $\Sigma$ have the property $\mathrm{F}_2$ of Theorem~\ref{thm:friedman} for $r=2$ with some constant $C$. By Theorem~\ref{thm:friedman} this assumption holds true with probability that tends to 1 as $n\to\infty$.

There exists two different states $q_1,q_2\in Q$ such that $q_1{\cdot}a=q_2{\cdot}a$. The set $Q{\cdot}a$ contains less than $n$ elements. If $|Q{\cdot}a|=1$, then $a$ is a reset word for $\mathrsfs{A}$. If $|Q{\cdot}a|>1$, take two different states $p_1,p_2\in Q{\cdot}a$. By $\mathrm{F}_2$, there is a product $w$ of less than $C\log n$ of the permutation letters such that $p_i{\cdot}w=q_i$ for $i=1,2$. Now
$|Q{\cdot}awa|<|Q{\cdot}a|$. If $|Q{\cdot}awa|=1$, $awa$ is a reset word for $\mathrsfs{A}$. If $|Q{\cdot}awa|>1$, we apply the same argument to a pair of different states in \fg{ $Q{\cdot}awa$.} Clearly, the process results in a reset word in at most $n-1$ steps while the suffix appended at each step is \fg{of length at most $C\log n+1$}. Hence the length of the reset word constructed this way is at most \fg{$(C+1)n \log n$}.\qed
\end{proof}

Corollary~\ref{cor:random} indicates that one can hardly discover an $n$-state automaton with the transition monoid equal to $T_n$ and sufficiently large reset threshold by a random sampling. Therefore, we performed an exhaustive search among all automata with two permutation letters generating $S_n$ and one letter of rank $n-1$. Our computational results are summarized in Table~1.

\begin{center}
\begin{tabular}{|C{3cm}|C{1cm}|C{1cm}|C{1cm}|C{1cm}|C{1cm}|C{1cm}|}
\hline
Number of states & 2 & 3 & 4 & 5 & 6 & 7 \\
\hline
Reset threshold & 1 & 4 & 8 & 14 & 19 & 27\\
\hline
\end{tabular}
\end{center}
\noindent Table 1: The largest reset thresholds of $n$-state automata two permutation letters generating $S_n$ and one letter of rank $n-1$

\medskip

As $n$ grows, the reset thresholds of the obtained examples become much smaller than $(n-1)^2$. We were unable to derive a series of $n$-state three-letter automata with the transition monoid $T_n$ and quadratically growing reset thresholds. We suspect that the reset threshold of automata in this class is $o(n^2)$.

In the case of unbounded alphabet, for every $n$, we present an $n$-state automaton $\mathrsfs{V}_n$ with the transition monoid $T_n$ such that $\rt(\mathrsfs{V}_n)=\tfrac{n(n-1)}{2}$. The state set of $\mathrsfs{V}_n$ is $Q_n=\{q_0, \dots, q_{n-1}\}$ and the input alphabet  consists of $n$ letters $a_1,\dots, a_n $. The transition function is defined as follows:
\begin{equation*}
     \left\{
     \begin{array}{ll}
     q_i{\cdot}a_j=q_i & \text{ for } 0\leq i<n,\ 1\leq j <n,\ i \neq j,\ i\neq j+1,\ j\neq n,\\
     q_i{\cdot}a_i=q_{i-1} & \text{ for } 0< i\leq n-1,\\
     q_{i}{\cdot}a_{i+1}=q_{i+1} & \text{ for } 0\leq i< n-1, \\
     q_0{\cdot}a_n=q_1{\cdot}a_n=q_0,\\
     q_i{\cdot}a_n=q_i & \text{ for } 2\leq i\leq n-1.
     \end{array}
     \right.
\end{equation*}
Simply speaking, every letter $a_i$ for $i \le n-1$ swaps the states $q_i$ and $q_{i-1}$ and fixes the other states. The letter $a_n$ brings both $q_0$ and $q_1$ to $q_0$ and fixes the other states. The automaton $\mathrsfs{V}_5$ is depicted in Fig.~\ref{Family5N}.

\begin{figure}[ht]
\begin{center}
\scalebox{0.95}{
\begin{tikzpicture}[->,>=stealth',shorten >=1pt,auto,node distance=2cm,
                    semithick]
  \tikzstyle{every state}=[fill=light-gray,draw=none,text=black, scale=1]

  \node[state] (A)                    {$q_0$};
  \node[state]         (B) [ right of=A] {$q_1$};
  \node[state]         (C) [right of=B] {$q_2$};
  \node[state]         (D) [right of=C] {$q_3$};
  \node[state]         (E) [right of=D] {$q_4$};

  \path (A) edge [loop above]  	node { $a_2$, $a_3$, $a_4$, $a_5$} (A)
  			edge    [bend left=10]          	node {$a_1$} (B)
        (B) edge   [bend left=10]	node {$a_1$, $a_5$} (A)
            edge    [bend left=10]          	node {$a_2$} (C)
            edge [loop above]  	node {$a_3$, $a_4$} (B)
        (C) edge      [bend left=10]        	node {$a_2$} (B)
        	edge      [bend left=10]        	node {$a_3$} (D)
         	edge [loop above] 	node {$a_1$, $a_4$, $a_5$} (D)
        (D) edge  		[bend left=10]	  	node {$a_3$} (C)
        	edge       [bend left=10]       	node {$a_4$} (E)
        	edge [loop above]		  	node {\fg{$a_1$, $a_2$, $a_5$}} (A)
		(E) edge  		[bend left=10]	  	node {$a_4$} (D)
        	edge [loop above]		  	node {$a_1$, $a_2$, $a_3$, $a_5$} (A);
\end{tikzpicture}}
\end{center}
\caption{The automaton $\mathrsfs{V}_5$}
\label{Family5N}
\end{figure}

Recall that a state $z$ of an DFA is said to be a \emph{sink state} (or \emph{zero}) if $z{\cdot}a=z$ for every input letter $a$. It is known that every $n$-state synchronizing automaton with zero can be reset by a word of length $\frac{n(n-1)}{2}$, cf.~\cite{Ry97}. To show that this upper bound is tight for each $n$, Rystsov~\cite{Ry97} constructed an $n$-state and $(n-1)$-letter synchronizing automaton $\mathrsfs{R}_n$ with zero which cannot be reset by any word of length less than $\frac{n(n-1)}{2}$. In fact, our automaton $\mathrsfs{V}_n$ is a slight modification of $\mathrsfs{R}_n$ as the latter automaton is nothing but $\mathrsfs{V}_n$ without the letter $a_1$.

\begin{theorem}
\label{ResetVn}
For every $n$, the automaton $\mathrsfs{V}_n$ has $T_n$ as its transition monoid and $\rt(\mathrsfs{V}_n)=\tfrac{n(n-1)}{2}$.
\end{theorem}
\begin{proof}
The letters $a_1,\ldots,a_{n-1}$ generate $S_n$ because the product $a_1\cdots a_{n-1}$ is a full cycle and any full cycle together with any transposition generates $S_n$. Since the letter $a_n$ has rank $n-1$, it together with $a_1,\ldots,a_{n-1}$ generates $T_n$.

The automaton $\mathrsfs{V}_n$ is synchronizing because so is the restricted automaton $\mathrsfs{R}_n$, and $\rt(\mathrsfs{V}_n)\le\tfrac{n(n-1)}{2}$ because every reset word for $\mathrsfs{R}_n$ resets $\mathrsfs{V}_n$ as well. It remains to verify that the length of any reset word for $\mathrsfs{V}_n$ must be at least $\tfrac{n(n-1)}{2}$. Let $w$ be a reset word of minimum length for $\mathrsfs{V}_n$. Since $a_n$ is the only non-permutation letter, we must have $w=w'a_n$ for some $w'$ such that $|Q_n{\cdot}w'|>1$. This is only possible when $Q_n{\cdot}w'=\{q_0,q_1\}$ whence $Q_n{\cdot}w=\{q_0\}$.
Consider the function $f$ from the set of all non-empty subsets of $Q_n$ into the set of non-negative integers defined as follows: if $S=\{q_{s_1}, \dots, q_{s_t}\}$, then $f(S)=\sum_{i=1}^t s_i$. Clearly, $f(\{q_0\})=0$ and $f(Q_n)=\tfrac{n(n-1)}{2}$. For any set $S$ and any letter $a_j$, we have $f(S{\cdot}a_j)\geq f(S)-1$ since each letter only exchanges two adjacent states or maps $q_1$ and $q_0$ to $q_0$. Thus, when we apply the word $w$ letter-by-letter, the value of $f$ after the application of the prefix of $w$ of length $i$ cannot be less than $\tfrac{n(n-1)}{2}-i$. Hence, to reach the value $0$, we need at least $\tfrac{n(n-1)}{2}$ letters.\qed
\end{proof}

\paragraph{\textbf{2.3. Upper \fg{Bound on the Reset Threshold.}}}

We now provide a quadratic upper bound on the reset words of automata with the transition monoid equal to $T_n$. Our proof is inspired by the method of Rystsov~\cite{rystsov2000estimation} adapted to our case.

Let $\mathrsfs{A}=\langle Q, \Sigma, \delta \rangle$ be a DFA. Given a proper non-empty subset $R\subset Q$ and a word $w$ over $\Sigma$, we say that $R$ \emph{can be extended by} $w$ \fg{if the cardinality of the set} $Rw^{-1}=\{q\in Q\mid q{\cdot}w\in R\}$ is greater than $|R|$. Now assume that $|Q|=n$ and the transition monoid of $\mathrsfs{A}$ coincides with the full transformation monoid $T_n$. Then there is a letter $x$ of rank $n-1$. The set $Q\setminus Q{\cdot}x$ consists of a unique state, which is called the \emph{excluded state} for $x$ and is denoted by $\excl(x)$. \fg{Furthermore}, the set $Q{\cdot}x$ contains a unique state $p$ such that $p=q_1{\cdot} x=q_2{\cdot}x$ for some $q_1\ne q_2$; this state $p$ is called the \emph{duplicate state} for $x$ and is denoted by $\dupl(x)$. We notice that a non-empty subset $R\subset Q$ can be extended by $x$ if and only if $\dupl(x)\in R$ and $\excl(x)\notin R$. Moreover, if a word $w$ \fg{is a product of permutation letters}, $R$ can be extended by the word $wx$ if and only if $\dupl(x)\in Rw^{-1}$ and $\excl(x)\notin Rw^{-1}$. To better understand which extensions are possible, we construct a series of directed graphs (digraphs) $\Gamma_i$, $i=0,1,\dotsc$, with the set $Q$ as the vertex set.

The digraph $\Gamma_0$ has the set $E_0 = \{(\excl(x),\dupl(x))\}$ as its edge set. Let $\Pi$ be the set of permutation letters of $\mathrsfs{A}$. Notice that $\Pi$ generates the symmetric group $S_n$. By $\Pi^{i}$ we denote the set of words of length at \fg{most $i$ over the letters in $\Pi$}. The digraph $\Gamma_i$ for $i>0$ has the edge set $E_i = \{(\excl(x){\cdot}w,\dupl(x){\cdot}w)\mid w\in\Pi^i\}$. The digraphs   $\Gamma_i$, $i=0,1,\dotsc$, form a sort of stratification for the graph $\Gamma_\infty$ with the edge set $E_\infty=\cup_{i=0}^\infty E_i$; the latter digraph has been studied  in~\cite{rystsov2000estimation} and~\cite{BondarVolkov16} (in the context of arbitrary completely reachable automata). Observe that none of the digraphs $\Gamma_i$, $i=0,1,\dotsc$, have loops.

Recall that a digraph is said to be \emph{strongly connected} if for every pair of its vertices, there exists a directed path from the first vertex to the second. \ffg{We need the two following lemmas.} 

\begin{lemma}
\label{lem:extension}
If the digraph $\Gamma_k$ is strongly connected, then every proper non-empty subset in $Q$ can be extended by a word of length at most $k+1$.
\end{lemma}
\jv{
\begin{proof}
let $R$ be a proper non-empty subset in $Q$. If $\Gamma_k$ is strongly connected, there exists an edge $(q,p)\in E_k$ that connects $Q\setminus R$ and $R$ in the sense that $q\in Q\setminus R$ while $p\in R$. As $(q,p)\in E_k$, there exists a word $w\in\Pi^k$ such that $(q,p)=(\excl(x){\cdot}w,\dupl(x){\cdot}w)$. Then $\dupl(x)\in Rw^{-1}$ and $\excl(x)\notin Rw^{-1}$, whence the word $xw$ extends $R$ and has length at most $k+1$.\qed
\end{proof}
}
\begin{lemma}
\label{lem:extension1}
The digraph $\Gamma_{2n-3}$ is strongly connected.
\end{lemma}
\jv{
\begin{proof}
We start with showing that the digraph $\Gamma_{n-1}$ contains an oriented cycle.

Consider the underlying digraph $\Delta$ of the automaton $\langle Q, \Pi, \delta|_{Q\times\Pi}\rangle$, i.e., the digraph with the vertex set $Q$ and the edge set $\{(q,q{\cdot}a)\mid q\in Q,\,a\in\Pi\}$. This digraph is strongly connected since $\Pi$ generates the whole symmetric group $S_n$. Therefore, for every $q\in Q{\cdot}x$, there exists a directed path in $\Delta$ from $\excl(x)$ to $q$. If one takes such a path $\excl(x)\xrightarrow{a_1}\cdots\xrightarrow{a_\ell}q$ of minimum length, it does not traverse any vertex in $Q$ more than once, whence the length $\ell$ of the path is at most $n-1$. Thus, the word $u=a_1\cdots a_\ell$ belongs to $\Pi^{n-1}$ and the pair $(\excl(x){\cdot}u,\dupl(x){\cdot}u)$ is an outgoing edge of the vertex $q=\excl(x){\cdot}u$ in the digraph $\Gamma_{n-1}$. We see that every state in $Q{\cdot}x$ has an outgoing edge in $\Gamma_{n-1}$. Now, we can walk along the edges of $\Gamma_{n-1}$, starting at $\excl(x)$, which has the outgoing edge $(\excl(x),\dupl(x))$, until we reach an already visited state, thus getting an oriented cycle in the graph.

If $\Gamma=(V,E)$ is a digraph, we say that a vertex $v'\in V$ is \emph{reachable} from a vertex $v\in V$ if either $v'=v$ or there is a directed path from $v$ to $v'$. The mutual reachability relation is an equivalence on the set $V$, and the digraphs induced on the classes of the mutual reachability relation are either \scn\ or singletons (i.e., digraphs with 1 vertex and no edge).  Slightly abusing terminology, we call these induced digraphs (including singletons) the \emph{\scn\ components} of the digraph $\Gamma$.

Consider the strongly connected components of the digraph $\Gamma_{n-1}$ and let $C_1,\dots,C_m$ denote their vertex sets. Without any loss we may assume that $|C_1|\ge|C_2|\ge\dots\ge|C_m|$. Observe that $m<n$ since $\Gamma_{n-1}$ contains an oriented cycle which is not a loop whence at least one \scc\ is non-singleton. (Recall that digraphs of the form $\Gamma_i$ are loopless.) If $m=1$, then already the digraph $\Gamma_{n-1}$ is strongly connected, and we are done. Otherwise we analyze the evolution of the partition of $\Gamma_k$ with $k\ge n-1$ into \scn\ components under the action of the letters in $\Pi$.
Since $\Pi$ generates the symmetric group $S_n$, it cannot preserve any non-trivial partition of $Q$. Thus, there is a non-singleton component $C$ among $C_1,\dots,C_m$ and a letter $a$ in $\Pi$ whose action sends two elements of $C$ to different components, i.e. $C{\cdot}a\cap C_i\ne\varnothing$ and $C{\cdot}a\cap C_j\ne\varnothing$ for some $C_i \neq C_j$.

By the definition of the sets $E_k$, if $(p,q)\in E_{n-1}$, then $(p{\cdot}a,q{\cdot}a)\in E_{n}$. Therefore each edge from $E_{n-1}$ that connects some vertices in $C_s$, $s=1,\dots,m$, translates into an edge from $E_n$ that connects the images of these vertices in $C_s{\cdot}a$. Therefore, the digraphs of $\Gamma_n$ induced on the sets $C_1{\cdot}a,\dots,C_m{\cdot}a$ are either \scn\ or singletons.
In particular, the digraph of $\Gamma_n$ induced on $C{\cdot}a$ is \scn. Since $C{\cdot}a\cap C_i\ne\varnothing$ and $C{\cdot}a\cap C_j\ne\varnothing$, the digraph of $\Gamma_n$ induced on the set $C{\cdot}a\cup C_i\cup C_j$ also is strongly connected. This implies that the number $m'$ of strongly connected components in $\Gamma_n$ is less than $m$. If $\Gamma_n$ is not yet \scn, the same reasoning applied to its \scn\ components, shows that the number of strongly connected components in $\Gamma_{n+1}$ is less than $m'$, etc.

Since at each step the number of strongly connected components is reduced at least by $1$, we conclude that we reach a \scn\ digraph in at most $n-2$ steps. Therefore, $\Gamma_{2n-3}$ is strongly connected.\qed
\end{proof}
}
\begin{theorem}
\label{ResetBabaiCerny}
Let $\mathrsfs{A}$ be an $n$-state automaton with the transition monoid equal to $T_n$. The reset threshold of $\mathrsfs{A}$ is at most $2n^2 - 6n + 5$.
\end{theorem}

\begin{proof}
Let $x$ be a letter of rank $n-1$ and $h=\dupl(x)$. We extend the set $\{h\}$ by $x$, getting a subset $R_2$ with $|R_2|\ge2$. Lemmas~\ref{lem:extension} and~\ref{lem:extension1} imply that proper non-empty subsets in $Q$ can be extended by words of length at most $2n-2$. Starting with $R_2$, we extend subsets until we reach the full state set. Let $u_i$ be the word of length at most $2n-2$ used for the $i$-th of these extensions and let $m$ be the number of the extensions. Observe that $m\le n-2$. Clearly, the word
$u_{m}\cdots u_1x$ resets  $\mathrsfs{A}$ and has the length at most $1+(n-2)(2n-2) = 2n^2 - 6n + 5$.\qed
\end{proof}

\begin{remark}
Let $\mathrsfs{A}=\langle Q, \Sigma, \delta \rangle$ be an $n$-state DFA that has a letter of rank $n-1$, and let $P$ be the subgroup of the symmetric group $S_n$ generated by the permutation letters from $\Sigma$. Our proof of Theorem~\ref{ResetBabaiCerny} actually works in the case if $P$ is a \emph{2-transitive} group, that is, $P$ acts transitively on the set of ordered pairs of $Q$.
\end{remark}

\section{Bounds \fg{on the Diameter of the Pair Digraph}}
\label{sec:2trans}
In this section we present a lower bound on the largest diameter of the pair digraph $P(A)$ for $A \subseteq S_n$. We proceed by presenting subsets $A \subseteq S_n$ for every odd $n$ whose diameter is $\tfrac{n^2}{4} + o(n^2)$. In order to simplify the presentation, we mostly use automata terminology and describe the corresponding examples as \fg{a family of automata} $\mathrsfs{F}_n=\langle Q_n, A, \delta \rangle$ (the input letters of $\mathrsfs{F}_n$ form the subset $A$). We let $Q_n=\{q_1,\dots,q_n\}$ and denote pairs of states such as $\{q_i,q_j\}$ simply by $q_iq_j$.

\begin{figure}[ht]
\begin{center}

\fg{
\scalebox{0.8}{
\begin{tikzpicture}[->,>=stealth',shorten >=1pt,auto,node distance=1.5cm,
                    semithick]
  \tikzstyle{every state}=[fill=light-gray,draw=none,text=black, scale=1]

   \node[state] (A)                    {$q_{1}$};
  \node[state]         (B) [above right of=A] {$q_{2}$};
  \node[state]         (C) [ below right of =B] {$q_{3}$};
  \node[state]         (D) [below left of =C] {$q_{4}$};
  \node[state]         (E) [right of =C] {$q_{5}$};
   \node[state]         (G) [left of =A] {$q_{6}$};

       \node   [state]       (K) [right of =E] {$q_{7}$};

  \path (A) edge           node [swap]{$a$} (B)
  			 edge      [bend left=20]       node {$b$} (G)
        (B) edge            node [swap] {$a$} (C)
        	edge       [loop left]       node {$b$} (B)
        (C) edge             node [swap]{$a$} (D)
        	edge     [bend left=20]         node {$b$} (E)
        (D) edge  	         node [swap]{$a$} (A)
       		edge    [loop left]          node {$b$} (A)

        (E) edge    [bend left=20]       node {$b$} (C)
        edge     [bend left=20]     node {$a$} (K)
        (G)edge    [bend left=20]       node {$b$} (A)
        edge    [loop left]          node {$a$} (A)
        (K)edge     [bend left=20]      node {$a$} (E)
        edge    [loop right]          node {$b$} (A);
\end{tikzpicture}
}
}
\end{center}
\caption{The automaton $\mathrsfs{F}_7$}
\label{VeryHardFamily7}
\end{figure}

The automaton $\mathrsfs{F}_7$ shown in Fig.~\ref{VeryHardFamily7} is the first of the family $\mathrsfs{F}_n$. The digraph of pairs of its states is shown in Fig.~\ref{VeryHardFamily7SQ}. One can verify that the shortest word mapping $q_2q_4$ to $q_4q_7$ has length 15. 

\begin{figure}[hb]
\begin{center}
\scalebox{0.8}{
\begin{tikzpicture}[->,>=stealth',shorten >=1pt,auto,node distance=1.5cm,
                    semithick]
  \tikzstyle{every state}=[fill=light-gray,draw=none,text=black, scale=1]

  \node[state] 		(A)                    {$q_{2}q_{4}$};
  \node[state]         (B) [left of=A] {$q_{1}q_{3}$};
  \node[state]         (C) [ below of =B] {$q_{5}q_{6}$};
  \node[state]         (D) [below of =C] {$q_{6}q_{7}$};

  \node[state]         (E) [right of =D] {$q_{1}q_{7}$};
  \node[state]         (F) [below right of =E] {$q_{2}q_{5}$};
  \node[state]         (G) [above right of =F] {$q_{3}q_{7}$};
  \node[state]         (H) [above left of =G] {$q_{4}q_{5}$};

  \node[state]         (I) [right of =G] {$q_{5}q_{7}$};

    \node[state]         (J) [right= 3.7cm of F] {$q_{2}q_{3}$};
  \node[state]         (K) [right = 3.7cm of H] {$q_{3}q_{4}$};
  \node[state]         (L) [right of =K] {$q_{1}q_{4}$};
  \node[state]         (M) [right of =J] {$q_{1}q_{2}$};

    \node[state]         (N) [right= 2cm of M] {$q_{2}q_{6}$};
  \node[state]         (O) [above left of =N] {$q_{1}q_{6}$};
    \node[state]         (P) [right = 2cm of L] {$q_{4}q_{6}$};
  \node[state]         (Q) [above right of =N] {$q_{3}q_{6}$};

 \node[state]         (R) [above of =Q] {$q_{1}q_{5}$};
  \node[state]         (S) [above left of =R] {$q_{2}q_{7}$};
    \node[state]         (T) [above right of =S] {$q_{3}q_{5}$};
  \node[state]         (U) [below right of =T] {$q_{4}q_{7}$};

  \path (A) edge     [bend left=20]     node {$a$} (B)
  			edge          [loop right]    node {$b$} (B)
         (B) edge    [bend left=20]     node {$a$} (A)
  			edge     [bend left=20]   node {$b$} (C)
            (C) edge    [bend left=20]     node {$a$} (D)
  			edge     [bend left=20]   node {$b$} (B)
		(D) edge    [bend left=20]     node {$a$} (C)
  			edge     [bend left=20]   node {$b$} (E)
        (E) edge        node {$a$} (F)
  			edge     [bend left=20]   node {$b$} (D)
        (F) edge        node {$a$} (G)
  			edge     [bend left=3]   node {$b$} (J)
        (G) edge        node {$a$} (H)
  			edge     [bend left=5]   node {$b$} (I)
       	(H) edge        node {$a$} (E)
  			edge     [bend left=3]   node {$b$} (K)
        (I) edge     [bend left=10]     node {$b$} (G)
  			edge          [loop right]    node {$a$} (B)
        (J) edge        node [swap]{$a$} (K)
  			edge     [bend left=3]   node {$b$} (F)
        (K) edge        node [swap] {$a$} (L)
  			edge     [bend left=3]   node {$b$} (H)
        (L) edge        node [swap]{$a$} (M)
  			edge     [bend left=5]   node {$b$} (P)
        (M) edge        node [swap]{$a$} (J)
  			edge     [bend left=5]   node {$b$} (N)
        (N) edge        node {$a$} (Q)
  			edge     [bend left=5]   node {$b$} (M)
        (O) edge        node {$a$}(N)
  			edge     [loop left]   node {$b$} (N)
        (P) edge        node {$a$} (O)
  			edge     [bend left=10]   node {$b$} (L)
       (Q) edge        node  {$a$} (P)
  			edge     [bend left=10]   node {$b$} (R)
      (R) edge        node [swap]{$a$} (S)
  			edge     [bend left=10]   node {$b$} (Q)
      (S) edge        node [swap]{$a$} (T)
  			edge     [loop left]   node {$b$} (Q)
      (T) edge        node [swap]{$a$} (U)
  			edge     [loop left]   node {$b$} (Q)
      (U) edge        node [swap]{$a$} (R)
  			edge     [loop right]   node {$b$} (Q);
\end{tikzpicture}
}
\end{center}
\caption{The pair digraph of $\mathrsfs{F}_7$}
\label{VeryHardFamily7SQ}
\end{figure}

The automata of the family are obtained recursively, starting with $\mathrsfs{F}_7$. From $\mathrsfs{F}_n$, we construct $\mathrsfs{F}_{n+2}$. The effect \fg{of the letters is the same} for the states $q_1, \dots, q_{n-2}$ in $\mathrsfs{F}_n$ and $\mathrsfs{F}_{n+2}$. The effect of the letters $a$ and $b$ at the states $q_{n-1}$, $q_n$, $q_{n+1}$ and $q_{n+2}$ is defined as follows: the letters mapping $q_{n-1}$ and $q_n$ to themselves in $\mathrsfs{F}_n$ exchange $q_{n-1}$ with $q_{n+1}$ and $q_{n}$ with $q_{n+2}$ respectively in $\mathrsfs{F}_{n+2}$. The other letter maps $q_{n+1}$ and $q_{n+2}$ to themselves and $q_{n-1}$, $q_{n}$ to $q_{n-3}$ and respectively $q_{n-2}$. The result is shown in Fig.~\ref{VeryHardFamily2k5odd} (for $n\equiv 3\pmod 4$), in which $k$ stands for $\tfrac{n-5}2$.

\fg{
\begin{figure}[ht]
\begin{center}
\scalebox{0.80}{
\begin{tikzpicture}[->,>=stealth',shorten >=1pt,auto,node distance=1.3cm,
                    semithick]
  \tikzstyle{every state}=[fill=light-gray,draw=none,text=black, scale=1]

   \node[state] (A)                    {$q_{1}$};
  \node[state]         (B) [above right of=A] {$q_{2}$};
  \node[state]         (C) [ below right of =B] {$q_{3}$};
  \node[state]         (D) [below left of =C] {$q_{4}$};  \node[state]         (E) [right of =C] {$q_{5}$};
   \node[state]         (G) [left of =A] {$q_{6}$};

       \node         (K) [right of =E] {$\cdots$};
       \node         (L) [left of =G] {$\cdots$};

       \node[state]         (F) [right of =K] {$q_{2k+3}$};
       \node[state]         (H) [right=0.4cm of F] {$q_{2k+5}$};
       \node[state]         (J) [left of =L] {$q_{2k+4}$};

  \path (A) edge              node [swap] {$a$} (B)
  			 edge      [bend left=20]       node {$b$} (G)
        (B) edge             node [swap] {$a$} (C)
        	edge       [loop left]       node {$b$} (B)
        (C) edge             node [swap] {$a$} (D)
        	edge     [bend left=20]         node {$b$} (E)
        (D) edge  	         node [swap] {$a$} (A)
       		edge    [loop left]          node {$b$} (A)

        (E) edge    [bend left=20]       node {$b$} (C)
        edge     [bend left=20]     node {$a$} (K)
        (F)edge    [bend left=20]       node {$b$} (K)
        edge      [bend left=20]      node {$a$} (H)
        (G)edge    [bend left=20]       node {$b$} (A)
        edge      [bend left=20]      node {$a$} (L)
        (H)edge    [bend left=20]       node {$a$} (F)
        edge     [loop right]      node {$b$} (E)
        (K)edge    [bend left=20]       node {$b$} (F)
        edge     [bend left=20]      node {$a$} (E)
        (L)edge    [bend left=20]       node {$b$} (J)
        edge     [bend left=20]      node {$a$} (G)
        (J)edge    [bend left=20]       node {$b$} (L)
        edge     [loop left]     node {$a$} (E);

\end{tikzpicture}
}
\end{center}
\caption{The automaton $\mathrsfs{F}_{2k+5}$, with $k$ odd}
\label{VeryHardFamily2k5odd}
\end{figure}
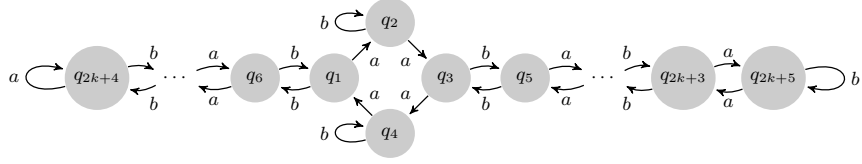
}

\begin{theorem}
\label{DiameterF}
For odd $n\geq 7$, the diameter of the pair digraph of the automaton $\mathrsfs{F}_n$ is at least $\tfrac{n^2}4+o(n^2)$.
\end{theorem}

\noindent\emph{Proof sketch.} For the automaton $\mathrsfs{F}_n$ ($n>7$, $n\equiv 3\pmod 4$), we claim that any word mapping $q_2q_4$ to $q_{k+2}q_{k+4}$ with $k=\tfrac{n-5}2$ has length at least $\tfrac{n^2}4+\tfrac{5n}4-7$. For this, we define a function $N$ which associates a non-negative integer $N(q_iq_j)$ to each pair $q_iq_j$, $i<j$. This function is such that if a pair $q_iq_j$ is mapped by $a$ or $b$ to a pair $q_{i'}q_{j'}$, then $N(q_{i'}q_{j'})\ge N(q_iq_j)-1$. This implies that if $(q_iq_j){\cdot}w=q_{s}q_{t}$ for some word $w$, then the length of $w$ is at least $N(q_iq_j)-N(q_{s}q_{t})$. The number assigned to $q_{k+2}q_{k+4}$ is 0, while the number given to $q_2q_4$ is equal to $\tfrac{n^2}4+\tfrac{5n}4-7$, thus, the claim holds. 

In addition, 
\jv{we describe} \fg{a word of length $\tfrac{n^2}4+\tfrac{5n}4-7$ that maps $q_2q_4$ to $q_{k+2}q_{k+4}$. Therefore $\tfrac{n^2}4+\tfrac{5n}4-7$ is the exact value of the ``distance'' between these two particular pairs}.

A similar argument holds for $n\equiv 1\pmod 4$, with the distance between two particular pairs of states being at least $\tfrac{n^2}4+\tfrac{5n}4-7.5$.\qed

\jv{
\begin{proof}
Recall that we aim to define a function $N$ which associates an integer $N(q_iq_j)\ge0$ to each pair $q_iq_j$, $i<j$, and has the following property:
\begin{equation}
\label{eq:subtraction}
\text{if $q_iq_j$ is mapped by $a$ or $b$ to a pair $q_{i'}q_{j'}$, then $N(q_{i'}q_{j'})\ge N(q_iq_j)-1$.}
\end{equation}
For an illustration, see Fig.~\ref{VeryHardFamily7SQN} which presents the pair digraph of the automaton $\mathrsfs{F}_7$ with the values of the corresponding function $N$ shown at each vertex.

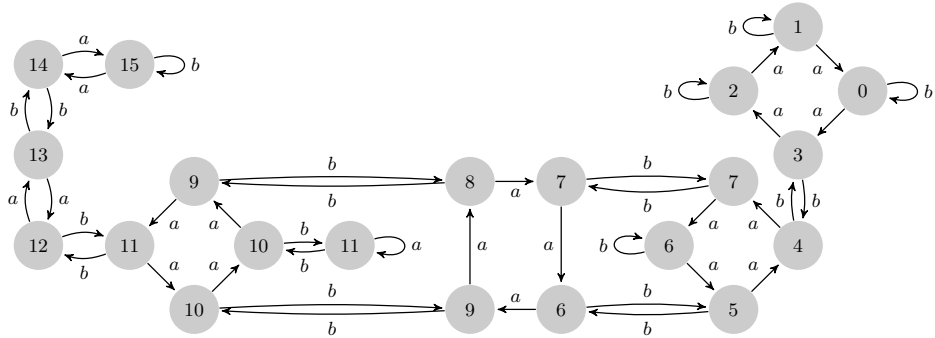
\begin{figure}[ht]
\begin{center}
\scalebox{0.8}{
\begin{tikzpicture}[->,>=stealth',shorten >=1pt,auto,node distance=1.5cm,
                    semithick]
  \tikzstyle{every state}=[fill=light-gray,draw=none,text=black, scale=1]

  \node[state] 		(A)                    {$15$};
  \node[state]         (B) [left of=A] {$14$};
  \node[state]         (C) [below of =B] {$13$};
  \node[state]         (D) [below of =C] {$12$};

  \node[state]         (E) [right of =D] {$11$};
  \node[state]         (F) [below right of =E] {$10$};
  \node[state]         (G) [above right of =F] {$10$};
  \node[state]         (H) [above left of =G] {$9$};

  \node[state]         (I) [right of =G] {$11$};

    \node[state]         (J) [right= 3.7cm of F] {$9$};
  \node[state]         (K) [right = 3.7cm of H] {$8$};
  \node[state]         (L) [right of =K] {$7$};
  \node[state]         (M) [right of =J] {$6$};

    \node[state]         (N) [right= 2cm of M] {$5$};
  \node[state]         (O) [above left of =N] {$6$};
    \node[state]         (P) [right = 2cm of L] {$7$};
  \node[state]         (Q) [above right of =N] {$4$};

 \node[state]         (R) [above of =Q] {$3$};
  \node[state]         (S) [above left of =R] {$2$};
    \node[state]         (T) [above right of =S] {$1$};
  \node[state]         (U) [below right of =T] {$0$};

  \path (A) edge     [bend left=20]     node {$a$} (B)
  			edge          [loop right]    node {$b$} (B)
         (B) edge    [bend left=20]     node {$a$} (A)
  			edge     [bend left=20]   node {$b$} (C)
            (C) edge    [bend left=20]     node {$a$} (D)
  			edge     [bend left=20]   node {$b$} (B)
		(D) edge    [bend left=20]     node {$a$} (C)
  			edge     [bend left=20]   node {$b$} (E)
        (E) edge        node {$a$} (F)
  			edge     [bend left=20]   node {$b$} (D)
        (F) edge        node {$a$} (G)
  			edge     [bend left=3]   node {$b$} (J)
        (G) edge        node {$a$} (H)
  			edge     [bend left=5]   node {$b$} (I)
       	(H) edge        node {$a$} (E)
  			edge     [bend left=3]   node {$b$} (K)
        (I) edge     [bend left=10]     node {$b$} (G)
  			edge          [loop right]    node {$a$} (B)
        (J) edge        node [swap]{$a$} (K)
  			edge     [bend left=3]   node {$b$} (F)
        (K) edge        node [swap] {$a$} (L)
  			edge     [bend left=3]   node {$b$} (H)
        (L) edge        node [swap]{$a$} (M)
  			edge     [bend left=5]   node {$b$} (P)
        (M) edge        node [swap]{$a$} (J)
  			edge     [bend left=5]   node {$b$} (N)
        (N) edge        node {$a$} (Q)
  			edge     [bend left=5]   node {$b$} (M)
        (O) edge        node {$a$}(N)
  			edge     [loop left]   node {$b$} (N)
        (P) edge        node {$a$} (O)
  			edge     [bend left=10]   node {$b$} (L)
       (Q) edge        node  {$a$} (P)
  			edge     [bend left=10]   node {$b$} (R)
      (R) edge        node [swap]{$a$} (S)
  			edge     [bend left=10]   node {$b$} (Q)
      (S) edge        node [swap]{$a$} (T)
  			edge     [loop left]   node {$b$} (Q)
      (T) edge        node [swap]{$a$} (U)
  			edge     [loop left]   node {$b$} (Q)
      (U) edge        node [swap]{$a$} (R)
  			edge     [loop right]   node {$b$} (Q) ;
\end{tikzpicture}
}
\end{center}
\caption{The pair digraph of $\mathrsfs{F}_7$, with function values}
\label{VeryHardFamily7SQN}
\end{figure}

For $n>7$, $n\equiv 3\pmod 4$, the values of the function $N$ are provided in the two lists below. Some of the formulas in the lists involve one or two positive integer parameters denoted by $m$ and $m'$. We always use $m'$ for the index of the first state of a pair and $m$ for the index of the second state; we do not specify the ranges of these parameters as they should be clear from the context. We use the following conventions: $k=\tfrac{n-5}2$, $N_1=\tfrac{k+3}2$, $N_2=(k+4)(k-1)$.

Our first list contains the values of $N$ for the pairs that involve one or two of the ``central'' states $q_1,q_2,q_3,q_4$ of the automaton $\mathrsfs{F}_n$ or one or two of its ``extreme'' states $q_{2k+4}$ and $q_{2k+5}$. (Our terminology follows the pictorial presentation of $\mathrsfs{F}_n$ in Fig.~\ref{VeryHardFamily2k5odd}.)
\begin{itemize}
\item $N(q_1q_2)=N_1+ N_2+2k+1$;
\item $N(q_1q_3)=N_1+N_2+4k+7$;
\item $N(q_1q_4)=N_1+N_2+2k+2$;
\item $N(q_1q_{4m+1})=\begin{cases}
\tfrac{k+3}2 &\text{if $m=N_1-1$},\\
N_1+(k+4)(k-2m-1)+2k+3 &\text{otherwise};
\end{cases}$
\item $N(q_1q_{4m+2})=N_1+(k+4)(k-2m-1)+2k-2m+2$;
\item $N(q_1q_{4m+3})=\begin{cases}
N_1+N_2+2k+6 &\text{if $m=1$},\\
N_1+(k+4)(k-2m+1)+2k+8&\text{otherwise};
\end{cases}$
\item $N(q_1q_{4m+4})=N_1+(k+4)(k-2m+1)+2k-2m+3$;
\item $N(q_1q_{2k+4})=N_1+k+2$;
\item $N(q_1q_{2k+5})=N_1+2k+8$;
\item[]

\item $N(q_2q_3)=N_1+N_2+2k+4$;
\item $N(q_2q_4)=N_1+N_2+4k+8$;
\item $N(q_2q_{4m+1})=\begin{cases}
N_1+N_2+2k+5 &\text{if $m=1$},\\
N_1+(k+4)(k-2m+1)+2k+7  &\text{otherwise};
\end{cases}$
\item $N(q_2q_{4m+2})=N_1+(k+4)(k-2m+1)+2k-2m+2$;
\item $N(q_2q_{4m+3})=N_1+(k+4)(k-2m-1)+2k+6$;
\item $N(q_2q_{4m+4})=N_1+(k+4)(k-2m-1)+2k-2m+1$;
\item $N(q_2q_{2k+4})=N_1+k+1$;
\item $N(q_2q_{2k+5})=N_1-1$;
\item[]

\item $N(q_3q_4)=N_1+N_2+2k+3$
\item $N(q_3q_{4m+1})=\begin{cases}
\tfrac{k-1}2 &\text{if $m=N_1-1$},\\
N_1+(k+4)(k-2m-1)+2k+5 &\text{otherwise};
\end{cases}$
\item $N(q_3q_{4m+2})=N_1+(k+4)(k-2m-1)+2k-2m+4$;
\item $N(q_3q_{4m+3})=\begin{cases}
N_1+N_2+2k+5 &\text{if $m=1$},\\
N_1+(k+4)(k-2m+1)+2k+6 &\text{otherwise};
\end{cases}$
\item $N(q_3q_{4m+4})=N_1+(k+4)(k-2m+1)+2k-2m+1$;
\item $N(q_3q_{2k+4})=N_1+k$;
\item $N(q_3q_{2k+5})=N_1+2k+6$;
\item[]

\item $N(q_4q_{4m+1})=\begin{cases}
N_1+N_2+2k+4 &\text{if $m=1$},\\
N_1+(k+4)(k-2m+1)+2k+5 &\text{otherwise};
\end{cases}$
\item $N(q_4q_{4m+2})=N_1+(k+4)(k-2m+1)+2k-2m+4$;
\item $N(q_4q_{4m+3})=N_1+(k+4)(k-2m-1)+2k+4$;
\item $N(q_4q_{4m+4})=N_1+(k+4)(k-2m-1)+2k-2m+3$;
\item $N(q_4q_{2k+4})=N_1+k+3$;
\item $N(q_4q_{2k+5})=N_1+1$;
\item[]

\item $N(q_{4m'+1}q_{2k+4})=\begin{cases}
N_1+N_2+3k+7 &\text{if $2m'=k+1$},\\
N_1+(k+4)(2m')+k+1 &\text{otherwise};
\end{cases}$
\item $N(q_{4m'+1}q_{2k+5})=\begin{cases}
N_1+(k+4)(2m'-2)+2k+4+2m' &\text{if $m'=N_1-1$};\\
N_1+(k+4)(2m'-2)+2k+5+2m' &\text{otherwise};
\end{cases}$
\item $N(q_{4m'+2}q_{2k+4})=N_1+(k+4)2m'+k+1+4m'$;
\item $N(q_{4m'+2}q_{2k+5})=\begin{cases}
N_1+2k+9 &\text{if $m'=1$},\\
N_1+N_2+2k+m'+8 &\text{if $2m'=k+1$},\\
N_1+(k+4)(2m'-2)+2k+2m'+7 &\text{otherwise};
\end{cases}$
\item $N(q_{4m'+3},q_{2k+4})=N_1+(k+4)(2m')+k$;
\item $N(q_{4m'+3}q_{2k+5})=\begin{cases}
N_1+(k+4)(2m')+2k+2m'+5 &\text{if $2m'=k-1$},\\
N_1+(k+4)(2m')+2k+2m'+6 &\text{otherwise};
\end{cases}$
\item $N(q_{4m'+4}q_{2k+4})=N_1+(k+4)2m'+k+2+4m'$;
\item $N(q_{4m'+4},q_{2k+5})=\begin{cases}
N_1+N_2+3k+5 &\text{if $2m'=k-1$},\\
N_1+(k+4)(2m')+2k+2m'+8  &\text{otherwise};
\end{cases}$
\item $N(q_{2k+4}q_{2k+5})=N_1+N_2+3k+6$.
\end{itemize}

Our second list contains the values of $N$ for the remaining pairs. In addition to our earlier conventions, we also use $M=m+m'$ and $M'=m-m'$ here.

\begin{itemize}
\item $N(q_{4m'+1}q_{4m+1})=N_1+(k+4)(k-2M'+1)+2k+2m'+5$;
\item $N(q_{4m'+1}q_{4m+2})=\begin{cases}
N_1+N_2+4k+8-2m &\text{if $m'=m$},\\
N_1+(k+4)(k-2M'+1)+2k-2m+2 &\text{otherwise};
\end{cases}$
\item $N(q_{4m'+1}q_{4m+3})=\begin{cases}
N_1+(k+4)(k-2M-1))+2k+2m'+4 &\text{if $2M<k+1$};\\
\tfrac{4m-k-1}2 &\text{if $2M=k+1$};\\
N_1+(k+4)(2M-k-3)+2k+2m'+5 &\text{otherwise};
\end{cases}$
\item $N(q_{4m'+1}q_{4m+4})=\begin{cases}
N_1+(k+4)(k-2M-1)+2k-2m+3 &\text{if $2M<k+1$};\\
N_1+2m  &\text{if $2M=k+1$};\\
N_1+(k+4)(2M-k-3) +2k+2m+8 &\text{otherwise};
\end{cases}$
\item[]

\item $N(q_{4m'+2}q_{4m+1})=\begin{cases}
N_1+N_2+2k+2m'+5 &\text{if $m'=m-1$},\\
N_1+(k+4)(k-2M'+1)+2k+2m'+7 &\text{otherwise};
\end{cases}$
\item $N(q_{4m'+2}q_{4m+2})= N_1+(k+4)(k-2M'+1)+2k-2m+4m'+2$;
\item $N(q_{4m'+2}q_{4m+3})=\begin{cases}
N_1+(k+4)(k-2M-1)+2k-2m+4 &\text{if $2M<k+1$};\\
N_1+2m'-1  &\text{if $2M=k+1$};\\
N_1+(k+4)(2M-k-3) +2k+2m'+7&\text{otherwise};
\end{cases}$
\item $N(q_{4m'+2}q_{4m+4})=\begin{cases}
N_1+(k+4)(k-2M-1)+2k-2m'+1 &\text{if $2M<k+1$};\\
N_1+k+2m'+1 &\text{if $2M=k+1$};\\
N_1+(k+4)(2M-k-1)+4m'+2m &\text{otherwise};
\end{cases}$
\item[]

\item $N(q_{4m'+3}q_{4m+1})=\begin{cases}
N_1+(k+4) (k-2M-1))+2k+2m'+5 &\text{if $2M<k+1$};\\
\tfrac{4m-k-3}2 &\text{if $2M=k+1$};\\
N_1+(k+4)(2M-k-3)+2k+2m'+6 &\text{otherwise};
\end{cases}$
\item $N(q_{4m'+3}q_{4m+2})=\begin{cases}
N_1+(k+4)(k-2M+1)+2k-2m+4 &\text{if $2M<k+1$};\\
N_1+2m-1  &\text{if $2M=k+1$};\\
N_1+(k+4)(2M-k-3) +2k+2m+7 &\text{otherwise};
\end{cases}$
\item $N(q_{4m'+3}q_{4m+3})=\begin{cases}
N_1+(k+4)(k-2M'+1)+2k+2m+3 &\text{if $m=m'+1$},\\
N_1+(k+4)(k-2M'+1)+2k+2m'+6 &\text{otherwise};
\end{cases}$
\item $N(q_{4m'+3}q_{4m+4})=\begin{cases}
N_1+N_2+4k+7-2m &\text{if $m'=m$},\\
N_1+(k+4)(k-2M'+1)+2k-2m+1&\text{otherwise};
\end{cases}$
\item[]

\item $N(q_{4m'+4}q_{4m+1})=\begin{cases}
N_1+(k+4)(k-2M-1)+2k-2m'+3 &\text{if $2M<k+1$};\\
N_1+2m'  &\text{if $2M=k+1$};\\
N_1+(k+4)(2M-k-3) +2k+2m'+8 &\text{otherwise};
\end{cases}$
\item $N(q_{4m'+4},q_{4m+2})=\begin{cases}
N_1+(k+4)(k-2M-1)+2k-2m'+1 &\text{if $2M<k+1$};\\
N_1+k+2m'+2 &\text{if $2M=k+1$};\\
N_1+(k+4)(2M-k-1)+4m+2m'-1 &\text{otherwise};
\end{cases}$
\item $N(q_{4m'+4}q_{4m+3})=\begin{cases}
N_1+N_2+2k+2m'+6 &\text{if $m=m'+1$},\\
N_1+(k+4)(k-2M'+1)+2k+2m'+8 &\text{otherwise};
\end{cases}$
\item $N(q_{4m'+4}q_{4m+4})=N_1+(k+4)(k-2M'+1)+2k-2m+4m'+3$.
\end{itemize}

It can be routinely verified that the function $N$ defined this way satisfies \eqref{eq:subtraction}. Therefore the shortest word mapping $q_2q_4$ to $q_{k+2}q_{k+4}$ is at least $N(q_2q_4)-N(q_{k+2}q_{k+4})$ letters long. Unfolding the definition, we obtain $N(q_{k+2}q_{k+4})=0$ and
\begin{multline*}
N(q_2q_4)=N_1+N_2+4k+8=\frac{k+3}2+(k+4)(k-1)+4k+8\\
=k^2+\frac{15k}2+\frac{11}2=\frac{n^2}4+\frac{5n}4-7.
\end{multline*}

\setcounter{table}{1}
In addition, Table \ref{word} provides the construction for a word of length $\frac{n^2}4+\frac{5n}4-7$ which maps $q_2q_4$ to $q_{k+2}q_{k+4}$. The word is composed of several factors being labels of a certain segments of the directed path from
$q_2q_4$ to $q_{k+2}q_{k+4}$ in the pair digraph of the automaton $\mathrsfs{F}_n$. For each segment we give its start and end pairs as well as its label and length. There are 4 ``starting'' factors of total length $4k+8$, followed by $2(k-1)$ ``inner'' factors, forming $\frac{k-1}2$ words of total length $2k+8$ each, and 2 ``finishing'' factors of total length $\frac{k+3}2$. Altogether they indeed form a word of length
\[
4k+8+k^2+3k-4+\dfrac{k+3}2=k^2+\frac{15k}2+\frac{11}2=\frac{n^2}4+\frac{5n}4-7.
\]

\begin{table}[h!]
\begin{center}
\begin{tabular}{|c| c|@\quad c@\quad|c|c|}
\hline
   &&&& Sum of lengths \\
   Start pair&End pair&Factor &Length of the factor&of factors\\
   &&&&within the group\\
   \hline
  $q_2q_4$& $q_1q_3$ & $a$ & 1& \\
  $q_1q_3$& $q_1q_7$ & $(ba)^kb$ & $2k+1$& \\
  $q_1q_7$& $q_3q_8$ & $aba^3ba$ & $7$& \\
  $q_3q_8$& $q_1q_{11}$ & $(ba)^{k-1}b$ & $2k-1$& \\
  &&&&$4k+8$\\
  \hline
  $q_1q_{11}$& $q_1q_5$ & $a^3ba\rule{0pt}{14pt}$ & $5$& \\
  $q_1q_5$& $q_3q_6$ & $b$ & $1$&\\
  $q_3q_6$& $q_3q_{12}$ & $a^3ba$ & $5$&\\
  $q_3q_{12}$& $q_1q_{15}$ & $(ba)^{k-2}b$ & $2k-3$& \\
    &&&&\\
    \hline
  $q_1q_{15}$& $q_1q_9$ & $a^3ba\rule{0pt}{14pt}$ & $5$& \\
  $q_1q_9$& $q_3q_{10}$ & $bab$ & $3$& \\
  $q_3q_{10}$& $q_3q_{16}$ & $a^3ba$ & $5$ &\\
  $q_3q_{16}$& $q_1q_{19}$ & $(ba)^{k-3}b$ & $2k-5$& \\
    &&&&\\
   \hline
   \dots & \dots & \dots & \dots & \\
   \hline
  $q_1q_{2k+1}$& $q_1q_{2k-5}$ & $a^3ba\rule{0pt}{14pt}$ & $5$&\\
  $q_1q_{2k-5}$& $q_3q_{2k-4}$ & $(ba)^{\frac{k-5}2}b$ & $k-4$ &\\
  $q_3q_{2k-4}$& $q_3q_{2k+2}$ & $a^3ba$ & $5$&\\
  $q_3q_{2k+2}$& $q_1q_{2k+5}$ & $(ba)^{\frac{k+1}2}b$ & $k+2$ &\\
      &&&&\\
  \hline
   $q_1q_{2k+5}$& $q_{1}q_{2k-1}$ & $a^3ba\rule{0pt}{14pt}$ & $5$&\\
   $q_{1}q_{2k-1}$& $q_{3}q_{2k}$ & $(ba)^{\frac{k-3}2}b$ & $k-2$&\\
   $q_{3}q_{2k}$& $q_{3}q_{2k+4}$ & $a^3ba$ & $5$&\\
   $q_{3}q_{2k+4}$& $q_1q_{2k+3}$ & $(ba)^{\frac{k-1}2}b$ & $k$ &\\
       &&&&$(2k+8)\dfrac{k-1}2\rule[4pt]{0pt}{14pt}$\\
       &&&&$=k^2+3k-4$\rule{0pt}{14pt}\\
   \hline
   $q_1q_{2k+3}$& $q_{3}q_{2k+3}$ & $a^2\rule{0pt}{14pt}$ & $2$ &\\
   $q_{3}q_{2k+3}$& $q_{k+2}q_{k+4}$ & $(ba)^{\frac{k-3}4}b$ & $\dfrac{k-1}2$ &\\
       &&&&$\dfrac{k+3}2$\rule[-8pt]{0pt}{18pt}\\
  \hline
  \end{tabular}
\end{center}
\caption{Construction of a word bringing $q_2q_4$ to $q_{k+2}q_{k+4}$}
\label{word}
\end{table}
\end{proof}
}

Our numerical experiments confirm that $\tfrac{n^2}4+\tfrac{5n}4-7$ is indeed the diameter of the pair graph of the automaton $\mathrsfs{F}_n$ for $n\equiv 3\pmod 4$ from $n=11$ to $n=31$ while $\tfrac{n^2}4+\tfrac{5n}4-7.5$ is the exact value of the diameter for $n\equiv 1\pmod 4$ from $n=13$ to $n=29$.

We have also computed the largest diameter of the pair digraph $P(A)$ for all $A \subseteq S_n$ with $|A|=2$ and $n=5,7,9$ and performed a number of random sampling experiments with two permutations for larger values of $n$. The experimental results suggest that the pair digraph of the automaton $\mathrsfs{F}_n$ has the largest diameter among all possible pair digraphs. Thus, we formulate the following:
\begin{conj}
The diameter of the pair digraph for a subset of $S_n$ \ffg{is bounded above} \fg{by $\tfrac{n^2}4 + o(n^2)$.}
\end{conj}

\section*{Conclusion}

We studied the hybrid Babai--\v{C}ern\'{y} problem, where the question is to find tight bounds on the reset threshold for automata with the full transition monoid. We presented a series of $n$-state automata $\mathrsfs{V}_n$ in this class with the reset threshold equal to $\tfrac{n(n-1)}2$, thus establishing a lower bound for the problem, and found an upper bound with the same growth rate, namely, $2n^2+o(n^2)$. We also described a series of $n$-state automata with diameter of the pair digraph equal to $\tfrac{n^2}4 + o(n^2)$.

For follow-up work, one direction is to refine the bounds with respect to the constants that do not match yet. Also, a lower bound for the hybrid problem using only three letters (generators) is of interest, since the number of letters of the presented family $\mathrsfs{V}_n$ is equal to the number of states.

\bibliography{gggjv}

\end{document}